\documentclass[a4paper,11pt]{article}
\usepackage{amsfonts,amsthm,amssymb}
\newcommand\COMP{\hbox{C\kern -.58em {\raise .54ex \hbox{$\scriptscriptstyle |$}}
\kern-.55em {\raise .53ex \hbox{$\scriptscriptstyle |$}} }}
\newcommand\NN{\hbox{I\kern-.2em\hbox{N}}}
\newcommand\RR{\hbox{I\kern-.2em\hbox{R}}}
\newcommand\sRR{{\it \hbox{I\kern-.2em\hbox{R}}}}
\newcommand\QQ{\hbox{I\kern-.53em\hbox{Q}}}
\newcommand\PP{\hbox{I\kern-.53em\hbox{P}}}
\newcommand\EE{\hbox{I\kern-.53em\hbox{E}}}
\newcommand\ZZ{{{\rm Z}\kern-.28em{\rm Z}}}
\newcommand\be{\begin{equation}}
\newcommand\ee{\end{equation}}
\newtheorem{theorem}{Theorem}[section]

\newtheorem{remark}[theorem]{Remark}

\newtheorem{example}[theorem]{Example}
\newtheorem{lemma}[theorem]{Lemma}

\newtheorem{definition}[theorem]{Definitions}

\def \Lbrack {[\![}
\def \Rbrack {]\!]}

\begin{document}
\title{How Non-Arbitrage, Viability and Num\'eraire Portfolio are Related\thanks{This research was supported financially by the Natural Sciences and Engineering Research Council of Canada, through Grant G121210818.}}



\author{ Tahir Choulli\thanks{corresponding author, {Email: tchoulli@ualberta.ca} }, Jun Deng and  Junfeng Ma\\ \\   Mathematical and Statistical Sciences Dept.\\ University of  Alberta, Edmonton, Alberta \\
}


\maketitle
\begin{abstract} This paper proposes two approaches that quantify the exact relationship among the viability, the absence of arbitrage, and/or the existence of the num\'eraire portfolio under minimal assumptions and for general continuous-time market models. Precisely, our first and principal contribution proves the equivalence among the
 No-Unbounded-Profit-with-Bounded-Risk condition (NUPBR hereafter), the existence of the num\'eraire portfolio, and the existence of the optimal portfolio under an equivalent probability measure for any ``nice" utility and positive initial capital. Herein, a ``nice" utility is any smooth von Neumann-Morgenstern utility satisfying Inada's conditions and the elasticity assumptions of Kramkov and Schachermayer. Furthermore, the equivalent probability measure ---under which the utility maximization problems have solutions--- can be chosen as close to the real-world probability measure as we want (but might not be equal). Without changing the underlying probability measure and under mild assumptions, our second contribution proves that the NUPBR is equivalent to the ``{\it local}" existence of the optimal portfolio. This constitutes an alternative to the first contribution, if one insists on working under the real-world probability. These two contributions lead naturally to new types of viability that we call weak and local viabilities.
\end{abstract}

  \vskip 0.15cm

\vskip 0.15cm


\section{Introduction} This paper discusses the three financial concepts of non-arbitrage, viability and num\'eraire portfolio. Among these concepts, the num\'eraire portfolio is the most recent concept that was introduced by Long in \cite{Long1990}. It is the portfolio with positive value process such that zero is always the best conditional forecast of the num\'eraire-dominated rate of return of every portfolio. The market's viability was defined ---up to our knowledge--- by Harrison and Kreps in \cite{Harrisonkreps} (see also \cite{kreps1981}, and \cite{Jouinikallalnapp2001}) as the setting for which there exists a risk-averse agent who prefers more to less, has continuous preference, and can find an optimal net trade subject to his budget constraint. In terms of the popular von Neumann-Morgenstern utility, the viability is essentially equivalent to the existence of the solution to the utility maximization problem. In contrast to the viability, the absence of arbitrage has several competing definitions which often vary with the market model in consideration. Among these, we can cite Non-Arbitrage, No-Unbounded-Profit-with-Bounded-Risk (NUPBR hereafter), No-Free-Lunch (NFL hereafter), No-Free-Lunch-with-Bounded-Risk, No-Free-Lunch-with-Vanishing-Risk (NFLVR hereafter), Asymptotic Arbitrage, Immediate Arbitrage, ..., etcetera. Philosophically, an arbitrage opportunity is a transaction with no cash outlay that results in a sure profit. In discrete-time markets, many arbitrage notions coincide and the following holds.
\begin{theorem}\label{finanacialEconomicalTheorem}For discrete-time market models with finite and deterministic horizons, the following are equivalent:\\
{\rm{(a)}} The market is viable/Utility Maximization admits solution for a "nice" von Neumann-Morgenstern utility,\\
{\rm{(b)}} Absence of arbitrage opportunities,\\
{\rm{(c)}} There exists an equivalent martingale measure (EMM hereafter),\\
{\rm{(d)}} The num\'eraire portfolio exists.
\end{theorem}

\noindent This theorem is a combination of results that were established, initially for discrete markets (i.e. markets with finite number of scenarios and trading times), in the economics and/or financial literature. In order to give precise references for these equivalences, we start with the equivalence between assertions (b) and (d), that was elaborated by Long in \cite{Long1990}, and was extended afterwards to general and different contexts by many scholars. For details, we refer the reader to the works of Artzner, Becherer, B\"uhlmann, Christensen, Karatzas, Kardaras, Korn, Larsen, Long, Platen, Sch\"al, and Sass (see \cite{Artzner}, \cite{DirkBecherer}, \cite{KaratzasKardaras2007}, \cite{korn2003}, \cite{ChristensenLarsen2007}, \cite{SassandSchal2012}, and the references therein). For this equivalence (i.e. equivalence between (b) and (d)), that is developed nowadays in full generality, our contribution lies in providing an easy proof.\\

\noindent The equivalence among (a), (b) and (c), to which our paper brings new ideas and original contribution, is termed in the financial literature as the Fundamental Theorem of Asset Pricing (FTAP hereafter) by Dybvig and Ross (see Theorems 1 and 2 of \cite{Dybvigross2003}). It is worth mentioning that this result, that goes back to Arrow and Debreu for discrete markets (see \cite{arrowdebreu} and \cite{duffine2008}), fails for the continuous-time setting and for the discrete-time case with infinite horizon as well. In mathematical finance the FTAP stands for the equivalence between (b) and (c), and for the rest of the paper this meaning will be adopted. This equivalence between (b) and (c) goes back to Kreps in \cite{kreps1981}, Harrison--Pliska in \cite{HarrisonPlisa}, and Dalang--Morton--Willinger in \cite{DalangMortonWillinger}. To obtain an analogous equivalence in the most general framework, Delbaen and Schachermayer had to strengthen the non-arbitrage condition (by considering NFLVR) while weakening the EMM (by considering $\sigma$-martingale measures). Their approach established the very general version of the FTAP in their seminal works \cite{DelbeanSchachermayer} and \cite{DelbeanSchachermayerBis}. The FTAP has been extended very successfully to examine markets with proportional transaction costs. The advancements in this direction stemming from the works of Guasoni, Jouini/Kallal, Kabanov, Rasonyi, and Schachermayer (see \cite{Kabanov}, \cite{JouiniKallal}, and \cite{Guasonirasonyischa2010} and the references therein) play a foundational r\^ole in the literature of mathematical finance.\\

\noindent The equivalence between (a) and (b) in discrete-time for smooth utilities was proved by \cite{Kallsen2001} and \cite{RasonyiStettner}. The utility maximization problem has been intensively investigated, under the assumption that (c) holds. This condition allows authors to use the two rich machineries of martingale theory and convex duality. These works can be traced back to \cite{karatzasshere1991}, \cite{frittelli2000}, \cite{Delbanen2002} and \cite{KramkovSchachermayer}, and the references therein to cite few. The main results in this literature focus on finding assumptions on the utility function for which duality can hold, and/or the solutions to the primal problem and its dual problem will exist.\\

\noindent The question of how the existence of optimal portfolio is connected to the absence of arbitrage (weak or strong form), in the continuous-time context, is treated nowhere up to now and up to our knowledge. Recently, Frittelli proposed in \cite{fritteli2007} an interesting approach for this issue, while his obtained results are not applicable in the context of \cite{Lowenstein2000} and \cite{ruf}.\\

\noindent  In this paper, we elaborate the equivalence among all four assertions of Theorem \ref{finanacialEconomicalTheorem} for the most general continuous-time framework under no assumption by choosing adequate notions and formulations. In particular, we prove that the NUPBR holds if and only if the optimal portfolio exists under an equivalent probability measure for any ``nice" utility and any positive initial capital. This main result together with other equivalent statements are detailed in Section 2, and are based on a technical lemma that is important in itself. This lemma closes the existing gap in the tight connection between (a) and (b) without changing the underlying probability measure. The proof and an extension of this lemma are given in Section 3.



\section{NUPBR, Weak Viability and Num\'eraire Portfolio}

 This section represents the core of the paper. In order to elaborate our main results, we start with describing the mathematical framework and formalizing mathematically the economic concepts used throughout the paper.  Our mathematical model is based on a filtered probability space $(\Omega, {\cal F }, \mathbb F, P)$, where the filtration, $\mathbb F:=({\cal F}_t)_{0\leq t \leq T}$, satisfies the usual conditions of right continuity and completeness. Here, $T$ is a finite and deterministic time horizon, and ${\cal F}_0$ is the completion of the trivial $\sigma$-field. On this stochastic basis, we consider a $d$-dimensional semi-martingale $(S_t)_{0\leq t\leq T},$ that represents the discounted price of $d$ risky assets. The space of martingales will be denoted by ${\cal M}(P)$, and the set of predictable processes that are $S$-integrable will be denoted by $L(S)$. The set ${\cal A}^+(Q)$ denotes the set of nondecreasing,
right-continuous with left-limits (RCLL hereafter), adapted and $Q$-integrable processes.\\

\noindent  The admissibility for strategies that will be used throughout the paper is given in the following.

\begin{definition}\label{admissibilityOfintegrands}
 Let $H=(H_t)_{0\leq t\leq T}$ be a predictable process.\\
 (i) For any positive constant $\alpha$, $H$ is called $\alpha$-admissible if $H$ is $S$-integrable and $(H\cdot S)_t \geq -\alpha,\ P-a.s$ for any $ t \in [0,T].$ \\
(ii) We say $H$ is admissible if there exists a positive constant $\alpha$ such that  $H$ is $\alpha$-admissible.
 \end{definition}
For any $x>0$, we define the set of wealth processes obtained from admissible strategies with initial capital $x$ by
\begin{equation}\label{admissiblewealth}
{\cal X}(x):=\left\{ X\geq 0\ \big|\ \mbox{there exists}\ H\in L(S)\ \mbox{such that}\ X=x+H\cdot S\ \right\}.\end{equation}

\begin{definition}\label{BK}
{\rm{(a)}} $S$ satisfies the NUPBR condition if the set
${\cal X}_T(1)$ is bounded in probability, where ${\cal X}_T(1)$ is the set of terminal values of elements of ${\cal X}(1)$.\\
{\rm{(b)}} $S$ satisfies non-arbitrage if
\begin{equation}\label{K(x)}
\{X_T\ \ \ \big|\ \ \ X\in \cup_{x>0}\left({\cal X}(x)-x\right)\}\cap L^0_+(P)=\{0\}.
\end{equation}
\end{definition}

\begin{definition}\label{sigmaDensity}
A $\sigma$-martingale density for $S$ is any positive local martingale, $Z$, such that there exists a real-valued predictable process $\phi$ satisfying $0<\phi\leq 1$ and $Z(\phi\cdot S)$ is a local martingale. The set of $\sigma$-martingale densities for $S$ will be denoted by
\begin{equation}\label{sigmamartingale}
{\cal Z}_{loc}(S):=\{Z\in{\cal M}_{loc}(P)\ \Big|\ Z_0=1,\ \ Z>0,\ \  ZS\ \mbox{is a $\sigma$-martingale}\ \}.
\end{equation}
\end{definition}

\begin{remark}\label{Remark1}
1) For any $Z\in {\cal Z}_{loc}(S)$ and any $H\in L(S)$ such that $H\cdot S\geq -\alpha$ $(\alpha\in{\mathbb R}^+)$,
 the process $Z(H\cdot S)$ is a local martingale. This follows immediately from Proposition 3.3  and
 Corollary 3.5 of \cite{anselstricker1994}.\\
 2) When the constant process one belongs to ${\cal Z}_{loc}(S)$, $S$ is called a $\sigma$-martingale. The notion of $\sigma$-martingale goes back to Chou \cite{chou1977} (see also \cite{emery}). It results naturally when we integrate --in the semimartingale sense-- an unbounded and predictable process with respect to a local martingale. The difference between $\sigma$-martingale and local martingale is discussed in \cite{anselstricker1994}.
\end{remark}

\noindent It is known from the literature that both concepts of viability and num\'eraire portfolio are involved with utility functions. Since the definition of utility is vague, herein we will work with ``nice" von Neumann-Morgenstern utilities. Below, we precise the mathematical definition of the utility and the corresponding admissible set of strategies afterwards.

\begin{definition}\label{Utility} A utility function is a function $U$ satisfying the following:\\
 (a) $U$ is continuously differentiable, strictly increasing, and strictly concave on its effective domain $\mbox{dom}(U)$.\\
 (b) There exists $u_0\in [-\infty,0]$ such that $\mbox{dom}(U)\subset (u_0,+\infty)$.\\
  The effective domain $\mbox{dom}(U)$ is the set of $r\in{\mathbb R}$ satisfying $U(r)>-\infty$.
\end{definition}
Given a utility function $U$, a semimartingale $X$, and a probability $Q$,  we define the set of admissible portfolios as follows
\begin{equation}\label{admissiblesetforU}\begin{array}{lll}
{\cal A}_{adm}(\alpha,U,X, Q) := \\ \left \{H\  |\ \  H \in L(X),\  H\cdot X \geq -\alpha\ \ \&\ \  E^Q\Bigl[U^{-}(\alpha + (H \cdot X)_T)\Bigr] < +\infty\right \}.\end{array}\end{equation}
When $Q=P$, $X=S$, and $U$ is fixed, we simply denote ${\cal A}_{adm}(\alpha,S)$.\\

\noindent Throughout this section, we will focus on utility functions $U$ satisfying
\begin{equation}\label{InadaConditiuons}
 \mbox{dom}(U)=(0,+\infty),\ \ U'(0) = +\infty,\ \ U'(\infty) = 0, \ \&\ \displaystyle\limsup_{x\rightarrow \infty}{{xU^{'}(x)}\over{U(x)}} < 1.
\end{equation}
These utilities were termed by ``nice" von Neumann-Morgenstern utilities in the abstract and the introduction. After recalling the mathematical definition of the num\'eraire portfolio, we will state our principal theorem of the paper, and will discuss its novelties by comparing it to the existing literature. Afterwards, we will provide its proof and other related technical results.

\begin{definition} Let $Q$ be a probability measure.
A process ${\widetilde X} \in {\cal X}(x_0)$ is called a num\'eraire portfolio under $Q$ if ${\widetilde X}>0$ and for every $X \in {\cal X}(x_0)$, the relative wealth process $X/{\widetilde X}$ is a $Q$-supermartingale .\\ If $Q=P$, then ${\widetilde X}$ is simply called the num\'eraire portfolio.
\end{definition}

\begin{remark} It is worth mentioning that our definition of num\'eraire portfolio, by abuse of terminology, assigns this terminology to the wealth process rather than the investment strategy generating this wealth as in \cite{KaratzasKardaras2007} on the one hand. On the other hand, herein, we do not assume the positivity of $\widetilde X_{-}$ as is the case of \cite{KaratzasKardaras2007}. However, this is not an extension in any way since this condition is always fulfilled under our definition above. In fact, for any $T\geq t>0$, thanks to Fatou's lemma, we have $E(1/{\widetilde X}_{t-})\leq x_0<+\infty$, and the positivity of ${\widetilde X}_{t-}$ follows.
\end{remark}

\subsection{The Main Result and Its Interpretations}
\noindent Below, we state the principal result of the paper.

\begin{theorem}\label{NumeraireLogarithm}
The following properties are equivalent:\\
(i) $S$ satisfies the NUPBR condition.\\
(ii) The set ${\cal Z}_{loc}(S)$ (defined in (\ref{sigmamartingale})) is not empty.\\
(iii) There exists a probability $Q\sim P$, such that for any utility $U$ satisfying (\ref{InadaConditiuons}) and any $x\in \mbox{dom}(U)$, there exists $\widehat\theta\in {\cal A}_{adm}(x,U,S,Q)$ such that
\begin{equation}\label{Umaximixation1}
\max_{\theta \in  {\cal A}_{adm}(x,U,S, Q)} E^Q U\Bigl(x + (\theta\cdot S)_{T}\Bigr) = E^QU\left(x + (\widehat{\theta}\cdot S)_{T}\right)< +\infty.
\end{equation}
(iv) For any $\epsilon >0$, there exists ${\widehat Q}({\epsilon})\sim P$ such that $E\vert {{d{\widehat Q}({\epsilon})}\over{dP}}-1\vert\leq \epsilon$, and for any utility $U$ satisfying (\ref{InadaConditiuons}) and any $x\in \mbox{dom}(U)$, there exists ${\widehat\theta}_{\epsilon}\in {\cal A}_{\epsilon,x}(U):={\cal A}_{adm}(x,U,S,{\widehat Q}({\epsilon}))$ such that
\begin{equation}\label{Umaximixation2}
\max_{\theta \in  {\cal A}_{\epsilon,x}(U)} E^{{\widehat Q}({\epsilon})} U\Bigl(x + (\theta\cdot S)_{T}\Bigr) = E^{{\widehat Q}({\epsilon})}U\left(x + ({\widehat\theta}_{\epsilon}\cdot S)_{T}\right)< +\infty.
\end{equation}
(v) For any $\epsilon\in (0,1)$, there exist ${\widetilde Q}({\epsilon})\sim P$ and ${\widetilde\theta}_{\epsilon}\in {\cal A}_{\epsilon,1}:={\cal A}_{adm}(1,\log, S, {\widetilde Q}({\epsilon}))$ such that  $E\vert {{d{\widetilde Q}({\epsilon})}\over{dP}}-1\vert\leq \epsilon$ and  \begin{equation}\label{logmaximixationattaun}
\max_{\theta \in {\cal A}_{\epsilon,1}} E^{{\widetilde Q}({\epsilon})}\log\Bigl(1 + (\theta\cdot S)_{T}\Bigr) = E^{{\widetilde Q}({\epsilon})}\log\left(1 + ({\widetilde\theta}_{\epsilon}\cdot S)_{T}\right) < +\infty.\end{equation}
(vi) The num\'eraire portfolio exists.
\end{theorem}

\noindent It is natural to ask how far and in which directions this theorem can be extended. Below, in the following remark, we will discuss two situations.

\begin{remark}\label{LarsenExample}
(a) It is important to mention that, in general, the NUPBR condition or equivalently the existence of the num\'eraire portfolio  does not guarantee the existence of the optimal portfolio under the real-world probability measure $P$. In fact in \cite{ChristensenLarsen2007}, the authors provide a model (Example 4.3 ) satisfying the NUPBR, while the log-utility maximization has no solution under $P$.\\
(b) Due to the importance of the exponential utility, it is quite natural to ask if we could extend Theorem \ref{NumeraireLogarithm} to the class of utilities with dom$(U)=\mathbb R$ (the case of real-valued wealth processes). We believe that the answer to this question is positive, while we prefer to keep our theorem in this form. Our main reason for this choice lies in the fact that extending Theorem \ref{NumeraireLogarithm} will certainly add technical complexity in the formulation itself. This will make our result difficult to interpret/understand and the key ideas of the theorem will be completely buried with technical conditions.
\end{remark}

\noindent In the remaining part of this subsection, we will discuss the original contribution of the theorem, its economic and financial interpretations, and its connection to the existing literature. The true novelty of this theorem lies in the equivalence among assertions (i), (iii), (iv), and (v). Before discussing the meanings of this innovation, we will first argue about the role of the remaining parts (i.e. (ii)$\Longleftrightarrow$(i)$\Longleftrightarrow$(vi)) that already exist in the literature.

\begin{remark}\label{partsNotNew} (a) The equivalence between (i) and (ii) is exactly Takaoka's result (see Theorem 2.6 in \cite{TAKAOKA}) on which our proofs rely heavily on the one hand. On the other hand ---for the reader's convenience--- by adding assertions (ii) and (vi), one can clearly see how Theorem \ref{finanacialEconomicalTheorem} becomes in the general continuous-time context under no assumption.\\
(b) The equivalence between (i) and (vi) was established for the first time ---up to our knowledge-- in Theorem 4.12 of \cite{KaratzasKardaras2007} (see also \cite{DirkBecherer}  and \cite{ChristensenLarsen2007}).  Our contribution here lies in the methodology used to prove this equivalence (see part 3) in Subsection \ref{SubsectionProofofmainTheorem}). It is worth mentioning that the proof for this equivalence in \cite{KaratzasKardaras2007} uses the semimartingale characteristics and the measurable selection theorem that are very powerful tools but not easy to handle. In contrast to \cite{KaratzasKardaras2007}, our approach uses Kolmos' argument, Fatou's lemma and the properties of the utility function only. Furthermore, our method constitutes an application of our original contribution (i)$\Longleftrightarrow$(v) ---which relies also on standard techniques---, and shows how to approximate the num\'eraire portfolio.
\end{remark}

 \noindent Theorem \ref{NumeraireLogarithm} can be interpreted from the financial/economic side and the mathematical finance side. Below, we will detail these two views.

 \begin{remark}
 a) From the mathematical finance perspective, our theorem suggests an alternative to the approaches of Becherer and  Christensen/Larsen (see \cite{DirkBecherer}  and \cite{ChristensenLarsen2007} and the references therein). In these works, the authors connected assertions (i) and (vi) to the existence of growth-optimal portfolio and to the existence of the solution to the log-utility maximization. A summary of these results is given by Hulley and Schweizer (see Theorem 2.3. of \cite{HardyHulleyMartinSchweizer2007}), where the authors stated that the assertions (i), (vi), and\\

\noindent (vii)\hskip .5cm  The growth-optimal portfolio $X^{go}$ exists,\\

\noindent are equivalent. If furthermore \begin{equation}\label{assumpionHally}
\sup\Bigl\{\ E\Bigl[\log X_T\Bigr]\  \Big|\ X\in {\cal X}(1),\ \ X_{-}>0,\ \mbox{and}\  E\left[(\log X_T)^{-}\right]< \infty\Bigr\} < \infty\end{equation}
holds, then assertions (i), (vi) and (vii) are also equivalent to:
\vskip 0.35cm

\noindent (viii)\hskip .5cm  The log-utility maximization problem admits a solution.\\

\noindent In Theorem \ref{NumeraireLogarithm}, we propose a new formulation for which the equivalence among the above four properties holds without any assumption and for any utility satisfying (\ref{InadaConditiuons}) --not only the log utility--. This formulation uses the appropriate change of probability. More importantly, the set of equivalent probabilities ---under which utilities satisfying (\ref{InadaConditiuons}) admit optimal portfolios--- is variation-dense. It is well known that the change of probability measure is a powerful probabilistic technique used in stochastic calculus to overcome integrability difficulties. Thus, mathematically speaking, the change of probability in Theorem \ref{NumeraireLogarithm} is a natural and adequate formulation that allowed us to establish the exact connection between the viability and the NUPBR under no assumption ---such as (\ref{assumpionHally})--- on the model. As mentioned in Remark \ref{LarsenExample}, in general, there is no hope for the existence of the optimal portfolio (even for the log utility) under the probability $P$. A curious reader can naturally ask what is the economical meaning of this probability change? To answer this question, we recall that in financial economics scholars called probability measures by agents' subjective believes. In this literature, the change of probability measures/believes has been well received and adopted since a while. The robust/uncertainty models and the random utility theory are among the successful areas of economics in which the change of probability is central. In this spirit of random utility theory, our assertion (iii) says that the market's viability is achieved by a random field utility for which (\ref{InadaConditiuons}) is fulfilled path-wisely. In mathematical terms, assertion (iii) is equivalent to\\
(iii') There exists a random field utility $\widetilde U(\omega,x)$ and a $\widetilde\theta\in{\cal A}_{adm}(x,\widetilde U)$ such that $\widetilde U(\omega,.)$ is a utility fulfilling (\ref{InadaConditiuons}) and
$$
\max_{\theta\in {\cal A}_{adm}(x,\widetilde U)}E\widetilde U\Bigl(x+(\theta\cdot S)_T\Bigr)=E\widetilde U\left(x+(\widetilde\theta\cdot S)_T\right).$$
For other situations, where the change of probability is economically motivated, and for the random utility theory literature, we refer the reader to \cite{choullima2013} and the references therein.\\
b) From the financial/economic view, our theorem is a generalization of Theorem \ref{finanacialEconomicalTheorem} to the most complex market model under no assumption. In fact, by substituting the viability under an equivalent belief and the NUPBR to assertions (a) and (b) of Theorem \ref{finanacialEconomicalTheorem} respectively, we obtained similar important result for continuous-time framework. Furthermore, our statement (iii) claims that any agent whose preference fulfills (\ref{InadaConditiuons}) can find optimal net trade under the same equivalent belief. This belief can be chosen as close to the real-world belief as we want (but might not be equal). This enhances our economic interpretation of the statement (iii) given by the following. \end{remark}

  \begin{definition}\label{weakviability}A market is weakly viable when there exist an agent ---whose utility fulfills (\ref{InadaConditiuons})--- and an initial capital for which the corresponding optimal portfolio exists under an equivalent probability measure. \end{definition}

\subsection{Proof of Theorem \ref{NumeraireLogarithm}}\label{SubsectionProofofmainTheorem}
\noindent The proof of this theorem is based essentially on three lemmas that we start with. The first lemma is dealing with the Fatou convergence of processes that was defined in Definition 5.2 of \cite{follmerandkramkoc1997}, while the second lemma deals with a supermartingale property. The third Lemma is the most important and innovative result among these three technical lemmas.

\begin{lemma}\label{ChoulliCorollary}
Suppose that ${\cal Z}_{loc}(S)\not=\emptyset$. Let $(\theta_n)_{n\geq 1}$ be such that $\theta_n\in L(S)$ and $\theta_n\cdot S\geq -1$. Then there exist $\phi_n\in \mbox{conv}(\theta_k,\ k\geq n)$ and $\widehat\theta\in L(S)$ and a nondecreasing process $C$ such that $\widehat\theta\cdot S\geq -1$, $C_0=0$, and
\begin{equation}\label{OD}
 1+\phi_n\cdot S\ \ \ \mbox{is Fatou convergent to}\ \ \  1+\widehat\theta\cdot S-C.\end{equation}
\end{lemma}
\begin{proof} The proof of the lemma follows immediately from combining Lemma 5.2 of \cite{follmerandkramkoc1997} and Theorem 2.1 of \cite{strickeryan1998}.
\end{proof}
\begin{lemma}\label{supermartingaleIntegral}
Let $X$ be any RCLL semimartingale, and $\widetilde\pi\in L(X)$ such that ${\cal E}(\widetilde\pi\cdot X)>0$. Then, the following are equivalent:\\
(i) For any $\pi\in L(X)$ such that ${\cal E}(\pi\cdot X)\geq 0$, and any stopping time, $\tau$, we have
\begin{equation}\label{inequalitysupermartingale}
E\Bigl[{{{\cal E}(\pi\cdot X)_{\tau}}\over{{\cal E}(\widetilde\pi\cdot X)_{\tau}}}\Bigr]\leq 1.\end{equation}
(ii) For any $\pi\in L(X)$ such that ${\cal E}(\pi\cdot X)\geq 0$, the ratio ${\cal E}(\pi\cdot X)/{\cal E}(\widetilde\pi\cdot X)$ is a supermartingale.
\end{lemma}
\begin{proof}
The proof of $(ii)\Longrightarrow (i)$ is obvious and will be omitted. Suppose that assertion $(i)$ holds, and consider $\pi\in L(X)$ such that ${\cal E}(\pi\cdot X)\geq 0$. Then, for any pair of stopping times, $\tau$ and $\sigma$, such that $\tau\leq \sigma\ P-a.s.$ and $A\in {\cal F}_{\tau}$, we put
$$
\overline\pi:=\widetilde\pi I_{\Rbrack 0,\tau_A\Rbrack}+\pi I_{\Rbrack\tau_A,+\infty\Lbrack},\ \ \ \tau_A:=\left\{\begin{array}{lll}\tau\ \ \ \ \ \ \ \ \ \mbox{on}\ A\\
+\infty\ \ \ \ \ \mbox{on}\ \ A^c\end{array}\right..$$
Then, we easily calculate

$${{{\cal E}(\overline\pi\cdot X)_{\sigma}}\over{{\cal E}(\widetilde\pi\cdot X)_{\sigma}}}={{{\cal E}(\widetilde\pi\cdot X)_{\tau}}\over{{\cal E}(\widetilde\pi\cdot X)_{\sigma}}}{{{\cal E}(\pi\cdot X)_{\sigma}}\over{{\cal E}(\pi\cdot X)_{\tau}}}I_A+I_{A^c}.
$$
Therefore, a direct application of (\ref{inequalitysupermartingale}) for $\overline\pi$ and $\sigma$, we obtain
$$
\displaystyle E\Bigl\{{{{\cal E}(\widetilde\pi\cdot X)_{\tau}}\over{{\cal E}(\widetilde\pi\cdot X)_{\sigma}}}{{{\cal E}(\pi\cdot X)_{\sigma}}\over{{\cal E}(\pi\cdot X)_{\tau}}}I_A\Bigr\}\leq P(A),$$
for any $A\in{\cal F}_{\tau}$. Hence, the supermartingale property for $\displaystyle{\cal E}(\pi\cdot X)\Bigl({\cal E}(\widetilde\pi\cdot X)\Bigr)^{-1}$ follows immediately, and the proof of the lemma is achieved.
\end{proof}

\begin{lemma}\label{positivewealththeorem1}
Let $U$ be a utility function satisfying (\ref{InadaConditiuons}). Suppose that there exists a sequence of stopping times $(T_n)_{n\geq 1} $ that increases stationarily to $T$ and $x_n > 0$ such that
\begin{equation}\label{finitenssofU1} \sup_{\theta \in  {\cal A}_{adm}(x_n, S^{T_n})} EU\Bigl(x_n + (\theta\cdot S)_{T_n}\Bigr) < +\infty, \ \ \ \ \ \  \forall\ n\geq 1. \end{equation}
Then, the following are equivalent:\\
(i) There exists a sequence of stopping times $(\tau_n)_{n\geq 1}$ that increases stationarily to $T$ such that for any $n\geq 1$ and any initial wealth  $x_0 >0$, there exists $\widehat{\theta}^{(n)} \in  {\cal A}_{adm}(x_0, S^{\tau_n})$  such that
\begin{equation} \max_{\theta \in  {\cal A}_{adm}(x_0, S^{\tau_n})} EU\Bigl(x_0 + (\theta\cdot S)_{\tau_n}\Bigr) = EU\Big(x_0 + (\widehat{\theta}^{(n)} \cdot S)_{\tau_n}\Bigr) < +\infty.  \end{equation}
(ii) $S$ satisfies the NUPBR condition.\\
\end{lemma}

\noindent This lemma will be interpreted economically, proved, and extended to the exponential utility in Section \ref{section11}. The remaining part of this section is devoted to the proof of Theorem \ref{NumeraireLogarithm}.

\begin{proof}{\it of Theorem \ref{NumeraireLogarithm}:} The proof of this theorem will be achieved after three steps. The first step will focus on proving  $(i) \Longleftrightarrow (iii)$. The second step will prove  $(i) \Longleftrightarrow (iv) \Longleftrightarrow (v)$, while the last step will address  $(v) \Longrightarrow (vi) \Longrightarrow (i)$.\\

\noindent {\bf Step 1)} The proof of $(i) \Longleftrightarrow (iii)$ boils down to the proof of (i) $\Longrightarrow$ (iii), since the reverse implication follows directly from Lemma \ref{positivewealththeorem1} by considering $Q$ instead of $P$ and taking $\tau_n=T_n=T$ for all $n\geq 1$. Suppose that assertion $(i)$ holds. Then, due to the equivalence between $(i)$ and $(ii)$, we consider $Z\in {\cal Z}_{loc}(S)$ (i.e. a $\sigma$-martingale density for $S$) and put
$$
 Q := \frac{Z_T}{E\left[Z_T\right]} \cdot P \sim P.
 $$
Let $U$ be a utility function satisfying (\ref{InadaConditiuons}) and $x\in \mbox{dom}(U)$. Thus, in virtue of Remark \ref{Remark1}, for any  $\theta\in {\cal A}_{adm}(x,U,S,Q)$, $Z(x + \theta\cdot S)$ is a nonnegative local martingale, and hence a supermartingale. Then, the concavity of $U$ leads to
\begin{equation}\label{maininequlaity}
E^QU\Bigl(x+(\theta\cdot S)_T\Bigr)\leq U\left({{x}\over{E[Z_T]}}\right)<+\infty,\ \mbox{for all}\ \theta\in {\cal A}_{adm}(x,U,S,Q).\end{equation}
Therefore, a direct application of Lemma \ref{positivewealththeorem1} under $Q$ implies the existence of a sequence of stopping times $(\tau_n)_{n\geq 1}$ that increases stationarily to $T$ and a sequence $\widehat\theta^{(n)}\in {\cal A}_{adm}(x, U,S^{\tau_n},Q)$ such that
\begin{equation}\label{maximilocal100}
 \sup_{\theta\in {\cal A}_{adm}(x,U,S^{\tau_n},Q)}E^QU\Bigl(x+(\theta\cdot S)_{\tau_n}\Bigr)=E^QU\left(x+(\widehat\theta^{(n)}\cdot S)_{\tau_n}\right).\end{equation}
 Thus, thanks to Lemma \ref{ChoulliCorollary}, we deduce the existence of $(a_l)_{l\geq 1}$ ($a_l\in (0,1))$, $\widehat\theta\in L(S)$, and a nondecreasing RCLL process $C$ such that $C_0=0$,
 \begin{equation}\label{Fatoulimit}
\displaystyle\sum_{l=n}^{m_n} a_l=1,\ \mbox{and}\  x+\sum_{l=n}^{m_n} a_l \widehat\theta^{(l)}\cdot S^{\tau_l}\ \ \ \mbox{is Fatou convergent to}\ \ \ x+\widehat\theta\cdot S-C. \end{equation}
Hence, assertion $(iii)$ will follow immediately once we prove that $\widehat\theta$ belongs to ${\cal A}_{adm}(x,U,S,Q)$ and it is the optimal solution to (\ref{Umaximixation1}). We start by proving the admissibility of $\widehat\theta$. Due to Fatou's lemma and the concavity of $U$, we get
\begin{equation}\label{firstineqaulity100}\begin{array}{lll}
  E^QU^-\Bigl(x+\widehat\theta\cdot S_T\Bigr)\leq \displaystyle\lim\inf_n E^QU^-\left(x+\sum_{l=n}^{m_n} a_l \widehat\theta^{(l)}\cdot S_{\tau_l}\right)  \\
  \\
  \hskip 3cm \leq  \displaystyle\lim\inf_n \sum_{l=n}^{m_n} a_l  E^QU^-\left(x+ \widehat\theta^{(l)}\cdot S_{\tau_l}\right).\end{array}\end{equation}
If $U(\infty)\leq 0$, then we have $$
\sum_{l=n}^{m_n} a_l  E^QU^-\left(x+ \widehat\theta^{(l)}\cdot S_{\tau_l}\right)=- \sum_{l=n}^{m_n} a_l  E^QU\left(x+ \widehat\theta^{(l)}\cdot S_{\tau_l}\right)\leq -U(x)<+\infty,$$
and the admissibility of $\widehat\theta$ follows immediately from this inequality and (\ref{firstineqaulity100}). Suppose that $U(+\infty)>0$. Then, there exists a real number $r$ such that $U(r) >0$, and the following hold
 $$\begin{array}{lll}
 \displaystyle \lim\inf_n \sum_{l=n}^{m_n} a_l  E^QU^-\left(x+ \widehat\theta^{(l)}\cdot S_{\tau_l}\right)\\
  \\
  \leq \displaystyle\lim\inf_n \sum_{l=n}^{m_n} a_l  E^QU\left(r+ x+ \widehat\theta^{(l)}\cdot S_{\tau_l}\right) - U(x)\leq U\left(\frac{r+x}{E[Z_T]}\right) - U(x) <+\infty.
\end{array}$$
 A combination of these inequalities and (\ref{firstineqaulity100}) completes the proof of $\widehat\theta\in {\cal A}_{adm}(x,U,S,Q)$. Furthermore, we get $U(x+\widehat\theta\cdot S_T)\in L^1(Q)$. Next, we will prove the optimality of the strategy $\widehat\theta$. To this end, we start by proving
 \begin{equation}\label{secondfinequality100}
 E^QU\left(x+\widehat\theta\cdot S_T\right)\geq \displaystyle\lim\sup_n E^QU\left(x+\sum_{l=n}^{m_n} a_l \widehat\theta^{(l)}\cdot S_{\tau_l}\right). \end{equation}
 If $U(+\infty)\leq 0$, then the above inequality follows from Fatou's lemma. Suppose that $U(+\infty)>0$. In this case, by mimicking the proof of Lemma 3.2 of \cite{KramkovSchachermayer}, we easily prove that
 \begin{equation}\label{unifoprmintegrability}
 \left\{ U(y_n):\ y_n:=x+\sum_{l=n}^{m_n} a_l \widehat\theta^{(l)}\cdot S_{\tau_l},\ n\geq 1\right\}\ \mbox{is $Q$-uniformly integrable}.\end{equation}
 Denote the inverse of $U$ by $\phi: (U(0+), U(+\infty)) \rightarrow (0,+\infty)$. Then we derive $E^Q[\phi(U(y_n))]\leq x/E(Z_T)$ and due to l'Hospital rule and (\ref{InadaConditiuons}) we have
 $$
 \lim_{x\rightarrow U(+\infty)} \frac{\phi(x)}{x} = \lim_{y\rightarrow +\infty} \frac{y}{U(y)} = \lim_{y\rightarrow +\infty} \frac{1}{U'(y)} =+\infty.
 $$
 Then, the uniform  integrability  of  the  sequence  $(U(y_n))_{n\geq 1}$  follows  from the La-Vall\'ee-Poussin argument. Then, (\ref{secondfinequality100}) follows immediately from this uniform integrability and (\ref{Fatoulimit}).
 Therefore, we obtain
\begin{eqnarray}
 E^QU\left(x+\widehat\theta\cdot S_T\right)&\geq& \displaystyle\lim\sup_n E^QU\Bigl(x+\sum_{l=n}^{m_n} a_l \widehat\theta^{(l)}\cdot S_{\tau_l}\Bigr)\nonumber\\
  &\geq& \displaystyle \lim\sup_n \sum_{l=n}^{m_n} a_l  E^QU\left(x+ \widehat\theta^{(l)}\cdot S_{\tau_l}\right)\nonumber\\
 &\geq& \displaystyle\lim\sup_n \sum_{l=n}^{m_n} a_l  E^QU\Bigl(x+ \epsilon\theta\cdot S_{\tau_l}\Bigr)\label{tobeexplain1}\\
 & \geq& \displaystyle \lim\inf_n\sum_{l=n}^{m_n} a_l  E^QU\Bigl(x+ \epsilon\theta\cdot S_{\tau_l}\Bigr)\nonumber\\
 \nonumber\\
 &\geq&\displaystyle E^QU\Bigl(x+ \epsilon\theta\cdot S_T\Bigr)\label{tobeexplain2}\\
\nonumber\\
 &\geq& \displaystyle(1-\varepsilon)U(x) + \varepsilon E^QU\Bigl(x+\theta\cdot S_T\Bigr),\nonumber
 \end{eqnarray}
for any $\theta\in {\cal A}_{adm}(x,U,S,Q)$, and any $\epsilon\in (0,1)$. It is clear that the optimality of $\widehat\theta$ follows immediately from the above inequalities by letting $\epsilon$ increases to one. It is obvious that (\ref{tobeexplain1}) follows from (\ref{maximilocal100}), while (\ref{tobeexplain2}) follows from Fatou's lemma and  $ U(x+\epsilon(\theta\cdot S)_{\tau_n})\geq U((1-{\epsilon})x)>-\infty$. This proves assertion $(iii)$, and the proof of $(i) \Longleftrightarrow (iii)$ is achieved.\\

\noindent{\bf Step 2)} Herein, we will prove $(i) \Longleftrightarrow (iv) \Longleftrightarrow  (v)$. Since the log-utility satisfies (\ref{InadaConditiuons}), then it is easy to see that the proof of $(i) \Longleftrightarrow (v)$ is similar to the proof of $(i) \Longleftrightarrow (iv)$. Thus, we will focus on proving this latter equivalence.\\

\noindent Suppose that assertion $(i)$ holds. Then, assertion $(iv)$ follows immediately as soon as we find $Q_{\delta}$ equivalent to $P$ whose density converges to one in $L^1(P)$ when $\delta$ goes to zero, and the utility maximization problem admits solution under $Q_{\delta}$ for any $\delta\in (0,1)$. To prove this latter claim, we put
\begin{equation}\label{variableq}
q:= \frac{Z_T}{E[Z_T]}, \ \ q_{\delta}:= {{q}\over{\delta + q}},\ \ Z_{\delta} := \frac{q_{\delta}}{E[q_{\delta}]} := q_{\delta}C_{\delta}, \ \ \ Q_{\delta}:=Z_{\delta}\cdot P\sim P,\end{equation}
for any $\delta\in (0,1)$. By examining closely the proof of $(i) \Longrightarrow (iii)$, we can easily conclude that the utility maximization problem admits solution under $Q_{\delta}$ whenever $Q_{\delta}$ satisfies similar inequality as in (\ref{maininequlaity}). Thus, for any utility $U$ satisfying (\ref{InadaConditiuons}), any $x\in dom(U)$, any $\delta\in (0,1)$, and any $\theta\in {\cal A}_{adm}(x,U,S,Q_{\delta})$, we derive
$$\begin{array}{llll}
E^{Q_{\delta}}U\Bigl(x+(\theta\cdot S)_T\Bigr)\leq U\Bigl(E^{Q_{\delta}}\left[x+(\theta\cdot S)_T\right]\Bigr)\leq\displaystyle U\Bigl({{E\left[Z_{T}[x+(\theta\cdot S)_T]\right]}\over{\delta E(q_{\delta})E(Z_T)}}\Bigr)\\
\\
\hskip 3.3cm \leq\displaystyle U\left({{x}\over{\delta E(q_{\delta})E[Z_T]}}\right)<+\infty.
\end{array}$$
Hence, this allows us to conclude that for any $\delta\in (0,1)$ and any utility $U$ satisfying (\ref{InadaConditiuons}), the utility maximization problem admits solution under $Q_{\delta}$. To conclude the proof of $(i) \Longrightarrow (iv)$, we will prove that $Z_{\delta}$ converges to one in $L^1(P)$ when $\delta$ goes to zero. Thanks to

$$
1>(C_{\delta})^{-1}=E\left({{q}\over{\delta + q}}\right) \geq E\left[\frac{q}{1+ q}\right]=:\Delta_0,$$
 we deduce that
$Z_{\delta}$ is positive, bounded by $(\Delta_0)^{-1}$, and converges almost surely to one when $\delta$ goes to zero. Then, for any $\epsilon>0$, the dominated convergence theorem implies the existence of $\delta:=\delta(\epsilon)>0$ such that $E\vert Z_{\delta(\epsilon)}-1\vert<\epsilon.$ This ends the proof of $(i) \Longrightarrow (iv)$. The reverse implication follows from $(iv) \Longrightarrow (iii)\Longrightarrow (i)$, and the proof of $(i) \Longleftrightarrow (iv) \Longleftrightarrow  (v)$ is completed. \\

\noindent This ends the proof of the original contributions of the theorem (i.e. $(i)\Longleftrightarrow (iii)\Longleftrightarrow (iv)\Longleftrightarrow (v)$). The remaining part of this proof ---as we explained before--- will prove the equivalence $(vi) \Longleftrightarrow (i)$  by applying $(i)\Longrightarrow (v)$, and using standard techniques such as Kolmos' arguments, Fatou's lemma and the properties of the utility function.\\

 \noindent{\bf Step 3)} In this last part, we will prove $(v) \Longrightarrow (vi) \Longrightarrow (i)$.  Suppose that assertion $(v)$ holds (and hence we have ${\cal Z}_{loc}(S)\not=\emptyset$). Then, it is easy to see that assertion $(v)$ implies the existence of the num\'eraire portfolio under each $Q_{\epsilon}$. Therefore, for any $n\geq 1$, there exist $0<Z_{n}={\kappa}_n q_n$ (here $q_n={{n}\over{n+q^{-1}}}$ where $q$ is given by (\ref{variableq})) that converges to one in $L^1(P)$, and $W^n$ the num\'eraire portfolio for $S$ under $Q_n:=Z_n\cdot P$. \\
 Hence, a direct application of Lemma \ref{ChoulliCorollary} leads to the existence of  $(\beta_n)_{n\geq 1}$ ($\beta_n\in (0,1)$), $\widetilde\theta\in L(S)$, and a nondecreasing and RCLL process $C$ such that $C_0=0$,
$$\sum_{k=n}^{m_n}\beta_l=1,\ \mbox{and}\ \sum_{k=n}^{m_n} \beta_k W^k\ \ \mbox{is Fatou convegent to}\ \ {\widetilde W}=x+\widetilde\theta\cdot S-C=:\widehat W-C.$$
Let $W\in {\cal X}(x)$ be a wealth process such that $W>0$, $b\in (0,1)$, $\alpha>1$, and $\tau$ be a stopping time. Then, there exists a sequence of stopping times $(\tau_k)_{k\geq 1}$ that decreases to $\tau$ and takes values in $\left(\hbox{I\kern-.53em\hbox{Q}}^+\cap [0,T[\right)\cup\{T\}$ such that
 $$\mbox{on}\ \ \ \{ \tau<T\}\ \ \ \ \ \ T\geq \tau_k>\tau,\ \ \ \mbox{ and on}\ \{\tau=T\} \ \ \ \ \tau_k=T.$$
 Due to Fatou's Lemma (using the convention ${{a}\over{0+}}=+\infty$, $a>0$), we obtain
$$\begin{array}{lll}
E\left({{W_{\tau}}\over{\widehat W_{\tau}}}\wedge \alpha\right)\leq E\left({{W_{\tau}}\over{\widetilde W_{\tau}}}\wedge \alpha\right)
\leq \displaystyle\lim\inf_{n}\lim\inf_{k} E\left(\displaystyle{{W_{\tau_k}}\over{\sum_{l=n}^{m_n} \beta_l W^l(\tau_k)}}\wedge \alpha\right)\\
\\
\hskip 4.4cm \leq \displaystyle\lim\inf_{n}\lim\inf_{k} E\Bigl(\left[\displaystyle\sum_{l=n}^{m_n} \beta_l{{W_{\tau_k}}\over{ W^l(\tau_k)}}\right]\wedge \alpha\Bigr).\end{array}$$
Since $q_n:={{n}\over{n+q^{-1}}}$ is increasing in $n$, then for any $l\geq n$ and any $k$ we have
$$\left\{ E(q_n|{\cal F}_{\tau_k})> b\right\}\subset\{ E(q_l|{\cal F}_{\tau_k})> b\}=\{1< b^{-1}{{Z_l(\tau_k)}\over{{\kappa}_l}}=E(q_l\big|{\cal F}_{\tau_k})b^{-1}\}.$$ Hence, we derive
$$\begin{array}{lll}
E\Bigl(\left[\displaystyle\sum_{l=n}^{m_n} \beta_l{{W_{\tau_k}}\over{ W^l(\tau_k)}}\right] \wedge\alpha\Bigr)=
E\Bigl(\left[\displaystyle\sum_{l=n}^{m_n} \beta_l{{W_{\tau_k}}\over{ W^l(\tau_k)}}\right]\wedge \alpha I_{\{ E(q_n|{\cal F}_{\tau_k})\leq b\}}\Bigr)+\\
\\
\hskip 4.5cm +
E\Bigl(\left[\displaystyle\sum_{l=n}^{m_n} \beta_l{{W_{\tau_k}}\over{ W^l(\tau_k)}}\right]\wedge \alpha I_{\{ E(q_n|{\cal F}_{\tau_k})> b\}}\Bigr)\\
\\
\hskip 4cm\leq  \alpha P\Bigl(E(q_n|{\cal F}_{\tau_k})\leq b\Bigr)+E\Bigl(\displaystyle\sum_{l=n}^{m_n} \beta_l{{Z_l(\tau_k)W_{\tau_k}}\over{b{\kappa}_l W^l(\tau_k)}} \Bigr)\\
\hskip 4cm\leq \alpha P\Bigl(E(q_n|{\cal F}_{\tau_k})\leq b\Bigr)+b^{-1}\displaystyle\sum_{l=n}^{m_n} {{\beta_l}\over{{\kappa}_l}}.
\end{array}$$
Since both ${\kappa}_n$ and $q_n$ converge to one when $n$ goes to infinity, and the random variable $E(q_n|{\cal F}_{\tau_k})$ converges to $E(q_n|{\cal F}_{\tau})$ when $k$ goes to infinity,  then it is obvious that $$\alpha P\Bigl(E(q_n|{\cal F}_{\tau_k})\leq b\Bigr)+b^{-1}\displaystyle\sum_{l=n}^{m_n} {{\beta_l}\over{{\kappa}_l}}\ \ \ \ \mbox{converges to}\ \ \ b^{-1},$$ when $k$ and afterwards $n$ goes to infinity. Hence, we deduce that
$$
E\left({{W_{\tau}}\over{\widehat W_{\tau}}}\wedge \alpha\right)\leq b^{-1},$$
for any $b\in (0,1)$, any $\alpha>1$, and any stopping time $\tau$. Thus, by taking $b$ to one and $\alpha $ to $+\infty$ and using Fatou's lemma, we deduce that
\begin{equation}\label{SupermgInequality}
E\left({{W_{\tau}}\over{\widehat W_{\tau}}}\right)\leq 1,\ \ \ \ \ \ \mbox{for any stopping time}\ \tau,\end{equation}
and for any positive wealth process $W\in {\cal X}(x)$. This proves that $\widehat W>0$, and (\ref{SupermgInequality}) holds for any $W\in {\cal X}(x)$ as well. Therefore, a combination of (\ref{SupermgInequality}) and Lemma \ref{supermartingaleIntegral} leads to the conclusion that $\widehat W$ is the num\'eraire portfolio under $P$. This completes the proof of assertion $(vi)$.\\

 \noindent The proof of the remaining implication (i.e. $(vi) \Longrightarrow (i)$) is easy, and will be detailed below for the sake of completeness. Suppose that there exists a num\'eraire portfolio $W^*$. Then, for any $\theta\in L(S)$ such that $1+\theta\cdot S\geq 0$, $$\frac{1 + \theta\cdot S}{W^*}\ \ \mbox{is a nonnegative supermartingale}.$$
As a result, for all $c>0$, we obtain
 $$\displaystyle P\Bigl(\frac{1 + (\theta\cdot S)_T}{W^*_T} > c\Bigr) \leq c^{-1} E\Bigl\{\frac{1 + (\theta\cdot S)_T}{W^*_T}\Bigr\} \leq c^{-1}.$$
 This clearly implies the boundedness of ${\cal X}_T(1)$ (the set of terminal values of the elements of ${\cal X}(1)$ defined in (\ref{admissiblewealth})) in probability and hence $S$ satisfies the NUPBR. This ends the proof of the theorem.  \end{proof}


\section{The Proof of Lemma \ref{positivewealththeorem1} and its Extension: Local Viability}\label{section11}

This section contains two subsections, where we prove Lemma \ref{positivewealththeorem1}, and develop its version for the exponential utility. These constitute our second main contribution of the paper. The condition (\ref{finitenssofU1}), in Lemma \ref{positivewealththeorem1}, is vital for the analysis of the utility maximization problem (see \cite{karatzasshere1991}, \cite{karazas2000}, and \cite{KramkovSchachermayer} and the references therein).  Furthermore, (\ref{finitenssofU1}) is irrelevant for the most innovative part of our lemma which is $(i) \Longrightarrow (ii)$. In the proof of Theorem \ref{NumeraireLogarithm}, where Lemma \ref{positivewealththeorem1} is applied, the condition (\ref{finitenssofU1}) is checked easily.  The reverse implication follows from the seminal work of Kramkov and Schachermayer (see \cite{KramkovSchachermayer}), and for the sake of completeness, details will be provided in the proof below. Below, in parts a) and b), we will discuss the meaning and the limitation of Lemma \ref{positivewealththeorem1} respectively.\\

\noindent{\bf a) What is the meaning of Lemma \ref{positivewealththeorem1}?} In virtue of Theorem \ref{NumeraireLogarithm}, Lemma \ref{positivewealththeorem1} proposes ---under assumption (\ref{finitenssofU1})--- an alternative to the equivalence between the NUPBR and the weak viability when working with the real-world probability measure is not an option. This lemma claims that, under mild assumptions, one can use the original belief $P$ and look for the optimal portfolio ``{\it locally}" instead of globally. The result of the lemma supports our definition of market's local viability as the market's viability up to a sequence of stopping times that increases stationarily to $T$ (respectively increases to infinity for the infinite horizon context). Furthermore, as mentioned in the introduction, this lemma closes the existing gap in quantifying the tightest relationship between the absence of arbitrage and the utility maximization {\it \` a la} Delbaen and Schachermayer (i.e. without changing measure, but by weakening and/or strengthening the concepts under consideration).\\

\noindent{\bf b) Can NFLVR be substituted into NUPBR in Lemma \ref{positivewealththeorem1}?} The stability of the NUPBR under the localization is a direct consequence of Takaoka's Theorem (see Theorem 2.6 in \cite{TAKAOKA}). In contrast to the NUPBR, Non-Arbitrage (see Definition \ref{BK}--(b)) or NFLVR (see \cite{DelbaenSchachermayer1994} and \cite{DelbeanSchachermayer} for its definition) can hold locally and fail globally. Thus, the existence of the optimal portfolio might not eliminate arbitrage opportunities in the model and hence NFLVR might be violated. For the sake of completeness, below we provide an example.

\begin{example}\label{MainExample}
Consider Example 4.6 of \cite{KaratzasKardaras2007}, where the market model is one stock on the finite time horizon $[0,1]$, with $S_0 =  1$ and $S$ satisfies $dS_t = (1/S_t)dt + d\beta_t.$ Here $\beta$ is a standard one-dimensional Brownian motion. It is worth mentioning that this example goes back to \cite{DelbaenSchachermayer1994}, and appeared in \cite{AnkirchnerImkeller} afterwards. In \cite{KaratzasKardaras2007}, the authors proved that both arbitrage opportunities and the num\'eraire portfolio (given by $\widetilde X=S$) exist for this model. Furthermore, it is easy to calculate
$$
\log({\cal E}(X)_1):=\log(S_1)=\int_0^1 {1\over{S_u}}d\beta_u+{1\over{2}}\int_0^1 {1\over{S_u^2}}du,\ \ \&\ \ E\int_0^1 {1\over{S_u^2}}du\leq 2\log(2).$$
Therefore, the log utility maximisation problem admits solution, while there is no equivalent martingale measure.
\end{example}

\subsection{Proof of Lemma \ref{positivewealththeorem1}:} We start by proving the easiest part of the lemma, which is $(ii) \Longrightarrow (i)$. Suppose that $S$ satisfies the NUPBR condition. Thanks to Takaoka's Theorem (see Theorem 2.6 in \cite{TAKAOKA}), we conclude the existence of a local martingale $Z>0$ and a real-valued predictable process $\varphi$ such that $0<\varphi\leq 1$ and $Z(\varphi\cdot S) \in {\cal M}_{loc}(P).$ Then, for any $\theta\in L(S)$ we have $\theta\cdot S=\theta^{\varphi}\cdot S^{\varphi}$ where $\theta^{\varphi}:=\theta/{\varphi}$ and $S^{\varphi}:=\varphi\cdot S$. Thus, without loss of generality, we assume that $ZS$ is a local martingale. Consider a sequence of stopping times, $(\sigma_n)_{n\geq 1}, $  that increases stationarily to $T$ such that both
$Z^{\sigma_n}$ and $Z^{\sigma_n}S^{\sigma_n}$ are martingales.  Put
$$Q_n := Z_{\sigma_n} \cdot P\ \ \ \ \mbox{and}\ \ \  \tau_n := T_n \wedge \sigma_n \uparrow T.$$
 Then, $Q_n$ is an equivalent martingale measure for $S^{\sigma_n}$. Since $\theta I_{\Lbrack 0,\tau_n\Rbrack}$ belongs to ${\cal A}_{adm}(x_n, S^{T_n})$ whenever $\theta \in  {\cal A}_{adm}(x_n, S^{\tau_n})$, for all $n\geq 1$ we derive
$$\sup_{\theta \in  {\cal A}_{adm}(x_n, S^{\tau_n})} EU\Bigl(x_n + (\theta\cdot S)_{\tau_n}\Bigr) \leq \sup_{\psi \in  {\cal A}_{adm}(x_n, S^{T_n})} EU\Bigl(x_n + (\psi\cdot S)_{T_n}\Bigr) < +\infty.$$
Therefore, a direct application of Theorems 2.1 and 2.2 of \cite{KramkovSchachermayer} implies that for any $n\geq 0$ and any initial wealth $x_0 >0$, there exists an $x_0$-admissible optimal strategy $\widehat{\theta}^{(n)}$ for $S^{\tau_n},$ such that
$$\max_{\theta \in  {\cal A}_{adm}(x_0, S^{\tau_n})} EU\Bigl(x_0 + (\theta\cdot S)_{\tau_n}\Bigr) = EU\Bigl(x_0 + (\widehat{\theta}^{(n)} \cdot S)_{\tau_n}\Bigr)<+\infty.$$
This proves assertion $(i)$. In the remaining part of the proof, we will focus on proving $(i) \Longrightarrow (ii)$. Suppose that assertion $(i)$ holds, and consider $x_0 = 1 + r$ such that  $r\in \mbox{dom}(U)$. Then, there exists $\widehat{\theta}^{(n)}\in {\cal A}_{adm}(1+r, S^{\tau_n})$ such that
$$\max_{\theta \in  {\cal A}_{adm}(1+r, S^{\tau_n})} EU\Bigl(1 + r + (\theta\cdot S)_{\tau_n}\Bigr) = EU\Bigl(1 + r + (\widehat{\theta}^{(n)} \cdot S)_{\tau_n}\Bigr) < +\infty.$$
For the sake of simplicity, we put $\tau := \tau_n$ and $\widehat{\theta}:=\widehat{\theta} ^{(n)}$ in what follows. In order to prove the NUPBR for $S^{\tau}$, we proceed by assuming that
$${\cal K} := \{(H\cdot S)_{\tau} | H \mbox {  is a 1-admissible strategy for } S^{\tau}\}$$
is not bounded in $L^0(P)$. Therefore, there exist a sequence of $1$-admissible strategy $(\theta^m)_{m\geq 1}$, a sequence of positive real numbers, $(c_m)_{m\geq 1}$, that increases to $+\infty$, and $\alpha > 0$ such that
$$P\Bigl((\theta^m\cdot S)_{\tau} \geq c_m\Bigr) > \alpha > 0.$$
Consider a sequence of positive numbers, $(\delta_m)_{m\geq 1}$, such that
$$0\leq \delta_m \rightarrow 0,\ \ \ \ \mbox{and}\ \ \ \ \delta_m c_m \rightarrow +\infty.$$
Then, put
$$X_m := \delta_m(\theta^m\cdot S)_{\tau} \geq -\delta_m,\ \ \ \ \ \  \mbox{for all}\ \ m \geq 1.$$
Hence, an application  of Kolmos's argument to $(X_m+\delta_m)_{m\geq 1}$ (see Lemma A1.1 of \cite{DelbaenSchachermayer1994}) leads to the existence of a sequence of random variables, $(g_k)_{k\geq 1}$, such that
$$ 0\leq g_k := \sum_{m=k}^{N_k}  \alpha_m X_m + \sum_{m=k}^{N_k} \alpha_m\delta_m \in \mbox{conv}\Bigl(X_m+\delta_m,\ m\geq k\Bigr) ,$$
and $g_k$ converges almost surely to $\widetilde{X} \geq 0$, with $P(\widetilde{X}>0) > 0.$\\
Since $y_k:= \displaystyle\sum_{m=k}^{N_k}  \alpha_m\delta_m$ converges to zero, we conclude that
$$-y_k \leq \widetilde{X}_k := \sum_{m=k}^{N_k}  \alpha_m\delta_m(\theta^m\cdot S)_{\tau}\ \ \mbox{converges to}\ \ \widetilde{X}\ \ \ P-a.s.,\ \ \mbox{and}$$
$$-(1+r)(1-y_k)\leq \widehat{X}_k := (1-y_k)(\widehat{\theta}\cdot S)_{\tau} \ \ \mbox{converges to}\ \ \ \ \ \ (\widehat{\theta}\cdot S)_{\tau}\ \ P-a.s .$$
Consider the new trading strategies
$$\widetilde{\theta}^{(k)} := \sum_{m=k}^{N_k}  \alpha_m\delta_m\theta^m + \Bigl(1 -  \sum_{m=k}^{N_k}  \alpha_m\delta_m\Bigr)\widehat{\theta} = \sum_{m=k}^{N_k}  \alpha_m\delta_m\theta^m + (1 -  y_k)\widehat{\theta}.$$
Then, it is easy to check that $1 + r +  \widetilde{\theta}^{(k)} \cdot S_{\tau} = 1 +  r + \widetilde{X}_k+ \widehat{X}_k \geq  y_k r > 0$ (due mainly to $-y_k \leq \widetilde{X}_k$ and $-(1+r)(1-y_k)\leq \widehat{X}_k$). Furthermore, due to the concavity of $U$, we have
\begin{eqnarray*}
U\Bigl(1 + r + (\widetilde{\theta}^{(k)} \cdot S)_{\tau}\Bigr) &=& U\Bigl(1 + r + \widetilde{X}_k+ \widehat{X}_k\Bigr)\\
&=& U\Bigl(1 + r + \widetilde{X}_k+ (1 -  y_k)(\widehat{\theta}\cdot S)_{\tau}\Bigr)\\
&\geq & U\Bigl(1 + r - y_k + (1 -  y_k)(\widehat{\theta}\cdot S)_{\tau}\Bigr)\\
&=& U\Bigl( y_kr + (1 -  y_k)\Bigl[1 + r + (\widehat{\theta}\cdot S)_{\tau} \Bigr]\Bigr)\\
&\geq& y_kU(r) + (1-y_k) U\Bigl(1 + r + (\widehat{\theta}\cdot S)_{\tau}\Bigr).
\end{eqnarray*}
This implies that $\widetilde{\theta}^{(k)} \in  {\cal A}_{adm}(1+r, S^{\tau})$. On the one hand, a combination of the previous inequality and Fatou's lemma implies that
\begin{equation}\label{general-utility-small}\begin{array}{llllllll}
E\left\{U\left(1 + r + \widetilde{X} + (\widehat{\theta} \cdot S)_{\tau}\right)- U\left(1 + r + (\widehat{\theta} \cdot S)_{\tau}\right) \right\}\\
\\
= E\left\{\displaystyle\lim_k \left[ U\Bigl(1 + r + \widetilde{X}_k + \widehat{X}_k\Bigr) - (1-y_k)U\Bigl(1 + r + (\widehat{\theta} \cdot S)_{\tau}\Bigr) - y_kU(r) \right]\right\} \\ \\
=E\left\{\displaystyle\lim_k \left[U\Bigl(1 + r + (\widetilde{\theta}^{(k)}\cdot S)_{\tau}\Bigr) - (1-y_k)U\Bigl(1 + r + (\widehat{\theta} \cdot S)_{\tau}\Bigr) - y_kU(r) \right]\right\} \\
\\
\leq \displaystyle\liminf_k E\left\{U\Bigl(1 + r + (\widetilde{\theta}^{(k)}\cdot S)_{\tau}\Bigr) - (1-y_k)U\Bigl(1 + r + (\widehat{\theta} \cdot S)_{\tau}\Bigr) -y_kU(r)\right\} \\
\\
\leq \displaystyle\liminf_k E\left\{U\Bigl(1 + r + (\widehat{\theta} \cdot S)_{\tau}\Bigr) - (1-y_k)U\Bigl(1 + r + (\widehat{\theta} \cdot S)_{\tau}\Bigr) -y_kU(r)\right\}\\
\\=0.
\end{array}\end{equation}

\noindent On the other hand, since $P(\widetilde{X}>0) > 0$ and $U$ is strictly increasing, we get
$$
E\left\{U\Bigl(1 + r + \widetilde{X} + (\widehat{\theta} \cdot S)_{\tau}\Bigr) \right\} > E\left\{U\Bigl(1 + r + (\widehat{\theta}\cdot S)_{\tau}\Bigr)\right\}.$$
This is a contradiction with (\ref{general-utility-small}), and the NUPBR for $S^{\tau}$ is fulfilled.  Then, the global NUPBR for $S$ is a direct consequence of Takaoka's Theorem (Theorem 2.6 of \cite{TAKAOKA}), and the proof of the theorem is completed.



\subsection{Extension to the case of exponential utility}

We believe that the extension of Lemma \ref{positivewealththeorem1} to the exponential utility is valuable and deserves attention for two reasons. The first reason lies in the popularity of the exponential utility, while the second reason lies in our belief that for this case, when $S$ is locally bounded, we may obtain more precise results with less assumptions. Throughout this section the set of admissible strategies for the model $(X,Q)$ will be denoted by $\Theta(X)$, and is given by
$$\Theta(X):=\displaystyle\left\{\theta\in L(X)\ \big|\ \ \theta\cdot X\ \mbox{are uniformly bounded in}\ (\omega,t)\  \right\}.$$
Then, the set of local martingale densities that are locally in $L\log L$ will be denoted by
\begin{equation}\label{setofZloffinite}
\mathcal{Z}_{f,loc}(X,Q):=\displaystyle\left\{Z>0\Big|\ Z,\ ZX\in{\cal M}_{loc}(Q),\ Z\log(Z)\ \mbox{is $Q$-locally integrable}\right\}.
\end{equation}
\noindent when $Q=P$, we simply write $\mathcal{Z}_{f,loc}(X)$.
\begin{definition}\label{MEH processes} Let $Z={\cal E}(N)\geq 0$, where $N\in {\cal M}_{0,loc}(P)$. If
\begin{equation}\label{VEntropy}
V^{(E)}(N):={1\over{2}}\langle N^c\rangle +\sum \Bigl[(1+\Delta N)\log(1+\Delta N)-\Delta N\Bigr],\end{equation}
is locally integrable, then its compensator is called the entropy-Hellinger process of $Z$ and is denoted by $h^{E}(Z,P)$ (see \cite{choullistricker2005} for details).
\end{definition}

\begin{lemma}\label{exponential-local-NUPBR}
Suppose $S$ is locally bounded. Then the following are equivalent: \\
(i) There exist  a sequence of stopping times $(\tau_n)_{n\geq1}$ increasing stationarily to $T$ and $\widehat\theta ^n\in L(S^{\tau_n})$ such that $\displaystyle E\left(\sup_{0\leq t\leq {\tau_n}}\exp\left[-(\widehat\theta^n\cdot S)_t\right]\right)<+\infty$ and
\begin{equation}\label{exp-localS-utilitymax}
\inf_{\theta\in\Theta(S^{\tau_n})} E\left(e^{-(\theta\cdot S)_{\tau_n}}\right)=E\left(e^{-(\widehat\theta^n\cdot S)_{\tau_n}}\right). \end{equation}
(ii) $\mathcal{Z}_{f,loc}(S)\neq \emptyset$.
\end{lemma}

\begin{proof}We start by proving $(ii)\Longrightarrow (i)$. Suppose that assertion $(ii)$ holds, and consider $Z\in \mathcal{Z}_{f,loc}(S)$. Then, there exists a sequence of stopping times, $(\tau_n)_{n\geq 1}$, that increases stationarily to $T$ such that $Z^{\tau_n}$ is a martingale and $h^E_{t\wedge{\tau_n}}(Z,P)$ is bounded. Therefore,  due to Theorem 3.7 or  Proposition 3.6 in \cite{choullistricker2005}, we deduce that $Q^n:=Z_{\tau_n}\cdot P$ is an equivalent martingale measure for $S^{\tau_n}$ satisfying the reverse H\"older condition $R_{L\log L}\left( P \right)$  (for the definition of reverse H\"older condition, we refer to \cite{Delbanen2002} or \cite{choullistricker2005}). Thus, Theorem 2.1 of \cite{KabanovStricker2002} implies the existence of the optimal solution $\widehat\theta^n\in L(S^{\tau_n})$ for (\ref{exp-localS-utilitymax}) such that $\exp\left[-(\widehat{\theta}^{n}\cdot S)_{\tau_n} \right]=E\exp\left[-(\widehat{\theta}^{n}\cdot S)_{\tau_n} \right] Z^{(E,n)}_{\tau_n}$ on the one hand. Here, $Z^{(E,n)}$ is the minimal entropy martingale density for $S^{\tau_n}$ which is an LlogL-integrable martingale and hence $E\left(\sup_{0\leq t\leq\tau_n} Z^{(E,n)}_t\right)<+\infty$. On the other hand, by Lemma 3.2 of \cite{Delbanen2002}, we conclude the existence of a positive constant $C_n$ such that $\exp\left[-(\widehat\theta^n\cdot S)_{t\wedge \tau_n}\right] \leq C_n Z^{E,n}_{t\wedge\tau_n}$. This ends the proof of assertion $(i)$.\\

\noindent In the remaining part of this proof, we will prove $(i)\Longrightarrow (ii)$. Suppose that assertion $(i)$ holds and put
\begin{equation}
 U^{(n)}_t := \exp\left( -\widehat{\theta}^{n}\cdot S_{t\wedge \tau_n} \right).
\end{equation}
Then by mimicking the proof of Lemma 4.1 in \cite{Delbanen2002}, we deduce that there exists a sequence of bounded strategies $(\theta^{(N)})_{N\geq 1} \subset \Theta\left( S^{\tau_n} \right)$ such that $P-a.s$
\begin{equation}\label{dominationinequality}
 \lim_{N \longrightarrow+\infty}e^{  - (\theta^{(N)}\cdot S)_{\tau_n} } =U^{(n)}_{\tau_n}\ \&\ \sup_{0\leq t\leq T} e^{  - \theta^{(N)}\cdot S_{t\wedge\tau_n} } \leq 6 \sup_{0\leq t\leq T} U^{(n)}_{t} \in L^1(P).
\end{equation}
Therefore, $\exp\left[- (\theta^{(N)}\cdot S)_{\tau_n}\right] $ converges to $U^{(n)}_{\tau_n}$ in $L^1$ when $N$ goes to $+\infty$.
For an arbitrary but fixed $\theta\in\Theta(S^{\tau_n})$ and any $\lambda\in(0,1)$, we denote
$$
\phi_{\lambda,N}:= -\lambda \theta+ \theta^{(N)} \in \Theta(S^{\tau_n}),
$$ and by making use of $Ee^{- \widehat{\theta}^n \cdot S_{\tau_n}}-Ee^{-\phi_\lambda \cdot S_{\tau_n}}\leq0 $ we derive
\begin{eqnarray*}
 \frac{E e^{-\theta^{(N)}\cdot S_{\tau_n}}-Ee^{-\phi_\lambda \cdot S_{\tau_n}}}{\lambda }
 &\leq& \frac{Ee^{-\theta^{(N)}\cdot S_{\tau_n}}-Ee^{-\widehat{\theta}^n \cdot S_{\tau_n}}}{\lambda } \rightarrow 0, \  \mbox{as } N\ \rightarrow +\infty.
\end{eqnarray*}
Due to (\ref{dominationinequality}) and $\theta\in \Theta(S^{\tau_n})$, the variable $(e^{-(\theta^{(N)}\cdot S)_{\tau_n}} -e^{-(\phi_{\lambda,N} \cdot S)_{\tau_n}})/{\lambda }$ converges to  $-( \theta\cdot S_{\tau_n})\exp[-\widehat{\theta}^n\cdot S_{\tau_n}]$ in $L^1(P)$  when $\lambda$ and $N$ go to zero and infinity respectively. By combining all the above remarks, we obtain
\begin{equation}
 E^{Q_n}\Bigl[ -(\theta\cdot S)_{\tau_n}\Bigr] \leq 0,\ \ \ \mbox{where}\ \ \ \  Q_n:={{\exp\left[-(\widehat{\theta}^n\cdot S)_{\tau_n}\right]}\over{E\left(\exp\left[-(\widehat{\theta}^n\cdot S)_{\tau_n}\right]\right)}}\cdot P.
\end{equation}
Since $\theta$ is arbitrary in $\Theta(S^{\tau_n})$, we conclude that $Q_n$ is an equivalent martingale measure for $S^{\tau_n}$. The density process of this martingale measure will be denoted by
$$\widehat{Z}^n_t:=\frac{E\left(\exp\left[-(\widehat\theta^n\cdot S)_{\tau_n}\right]\Big|\mathcal{F}_{t}\right)}{E\left(\exp\left[-(\widehat\theta^n\cdot S)_{\tau_n}\right]\right)}=:{\cal E}_t\left(\widehat N^{(n)}\right).$$
For any $\theta\in\Theta(S^{\tau_n})$,  and any $\lambda\in(0,1)$, on the one hand, the convexity of $e^x$ leads to conclude that $((\theta\cdot S)_{\tau_n}- (\widehat\theta^n\cdot S)_{\tau_n}) \exp(-(\widehat\theta^n\cdot S)_{\tau_n})$ is bounded from below by $-\exp(-(\theta\cdot S)_{\tau_n})\in L^1(P)$. On the other hand, again the convexity of $e^x$ combined with Fatou's lemma and the minimality of $\widehat{\theta}^n$ imply that
\begin{eqnarray*}
E\left(e^{-(\widehat\theta\cdot S)_{\tau_n}}((\theta-\widehat\theta^n)\cdot S)_{\tau_n}\right)&\leq&  \lim_{\lambda \rightarrow 0}E\left(e^{-\widehat\theta\cdot S_{\tau_n}}\frac{1 - \exp\left[-\lambda(( \theta - \widehat{\theta}^n)\cdot S)_{\tau_n}\right]}{\lambda}\right)\\
\\
 &\leq& 0,\end{eqnarray*}
This proves that $K_n:=(\widehat\theta^n\cdot S)_{\tau_n}\exp\left[{-(\widehat\theta^n\cdot S)_{\tau_n}}\right]\in L^1(P)$. By combining this with
$$\widehat{Z}^n_{\tau_n}\log(\widehat{Z}^n_{\tau_n})=\frac{-K_n-\exp(-(\widehat\theta^n\cdot S)_{\tau_n})\log\left(E\left[\exp(-(\widehat\theta^n\cdot S)_{\tau_n})\right] \right)}{E\left(\exp\left[-(\widehat\theta^n\cdot S)_{\tau_n}\right]\right)},$$
we deduce that $\widehat{Z}^n_{\tau_n}\log(\widehat{Z}^n_{\tau_n})$ is integrable, and hence $\widehat{Z}^n$ is a martingale density for $S^{\tau_n}$ that is $L\log L$-integrable. Then, by putting
$$ {\widehat N}:=\sum_{n=1}^{+\infty} I_{\Rbrack \tau_{n-1}, \tau_n\Rbrack}\cdot \widehat N^{(n)},$$
and applying Lemma \ref{MEHlocalization} below, assertion $(ii)$ follows immediately. This ends the proof of the lemma.\end{proof}

\begin{remark} The extension of Lemma \ref{positivewealththeorem1} to general utilities $U$ with dom$(U)=\mathbb R$ can be found in \cite{Deng2014}. This extension is less attractive, when comparing it to Lemmas \ref{positivewealththeorem1} and \ref{exponential-local-NUPBR}, due to the required technical assumptions.\end{remark}

\begin{lemma}\label{MEHlocalization}
Let $(\tau_n)_{n\geq 1}$ be a sequence of stopping times that increases stationarily to $T$, and $(N^{(n)})_n$ be a sequence of local martingales. Then, the process
$$
N:=\sum_{n=1}^{+\infty} I_{\Rbrack \tau_{n-1},\tau_n\Rbrack}\cdot N^{(n)}, \ \ (\tau_0 = 0),
$$
is a local martingale satisfying the following.\\
(i) If ${\cal E}(N^{(n)})>0$ for any $n\geq 1$, then  ${\cal E}(N)>0$.\\
(ii) If  $V^{(E)}(N^{(n)})\in {\cal A}^+_{loc}(P)$ for any $n\geq 1$, then  $V^{(E)}(N)\in {\cal A}^+_{loc}(P)$.\\
(iii) If ${\cal E}(N^{(n)})$ is a $\sigma$-martingale density for $S^{\tau_n}$ for any $n\geq 1$, then  ${\cal E}(N)$ is a $\sigma$-martingale density for $S$.
\end{lemma}

\begin{proof}It is obvious that
$$N^{\tau_n}=\sum_{k=1}^n I_{\Rbrack \tau_{k-1},\tau_k\Rbrack}\cdot N^{(k)}\in {\cal M}_{0,loc}(P).$$
This proves that $N\in \left({\cal M}_{0,loc}(P)\right)_{loc}={\cal M}_{0,loc}(P)$, and ${\cal E}(N)>0$ since
$$
1+\Delta N=1+\Delta N^{(n)}>0\ \ \ \mbox{on}\ \ \ \Rbrack \tau_{n-1},\tau_n\Rbrack,\ \ \ \ \ n\geq 1.$$
Then, due to the definition of the operator $V^{(E)}$ given by (\ref{VEntropy}), it is also easy to remark that $V^{(E)}(I_{\Rbrack \sigma,\tau\Rbrack}\cdot M)=I_{\Rbrack \sigma,\tau\Rbrack}\cdot V^{(E)}(M)$ for any local martingale $M$ (with $1+\Delta M\geq 0$) and any pair of stopping times $\tau$ and $\sigma$ such that $\tau\geq\sigma$. Thus, we get
$$
\left(V^{(E)}(N)\right)^{\tau_n}=\sum_{k=1}^n I_{\Rbrack \tau_{k-1},\tau_k\Rbrack}\cdot V^{(E)}(N^{(k)})\in {\cal A}^+_{loc}(P).$$
Hence, we deduce (thanks to Lemma 1.35 of \cite{JacodandShirayev2003}) that $V^{(E)}(N)\in\left({\cal A}^+_{loc}(P)\right)_{loc}={\cal A}^+_{loc}(P)$. This ends the proof of assertion (i) and (ii) of the lemma. To prove the last assertion, we first remark that ${\cal E}(M)$ is a $\sigma$-martingale density for $S$ if and only if there exists a predictable process $\varphi$ such that $0<\varphi\leq 1$ and
$$
\varphi\cdot S+\varphi\cdot [S,M]\in {\cal M}_{0,loc}(P).$$
Therefore, since ${\cal E}(N^{(n)})$ is a $\sigma$-martingale density for $S^{\tau_n}$ for each $n\geq 1$, then there exists $\phi_n$ such that $0<\phi_n\leq 1$ and
\begin{equation}\label{martingaleequation0}
Y_n:=\phi_n\cdot S+\phi_n\cdot [S^{\tau_n},N^{(n)}]\in {\cal M}_{0,loc}(P),\ \ \ \ \ \ \ \forall \ \ n\geq 1.\end{equation}
Put $\phi:=\sum_{k=1}^{+\infty} I_{\Rbrack \tau_{k-1},\tau_k\Rbrack}\phi_k$. Thus, it is easy to prove that $0<\phi\leq 1$, and
$$
\left(\phi\cdot S+\phi\cdot [S,N]\right)^{\tau_n}=\sum_{k=1}^{n} I_{\Rbrack \tau_{k-1},\tau_k\Rbrack}\cdot Y_k\in {\cal M}_{0,loc}(P).$$
Hence,
$
\phi\cdot S+\phi\cdot [S,N]\in \left({\cal M}_{0,loc}(P)\right)_{loc}={\cal M}_{0,loc}(P),$
 and hence ${\cal E}(N)$ is a $\sigma$-martingale density for $S$. This ends the proof of the lemma. \end{proof}

\textbf{Acknowledgements:} The first author is very grateful to Monique Jeanblanc, Kostas Kardaras, Miklos Rasonyi, Martin Schweizer, and Moris Strub for their useful comments and remarks that helped improving the paper and for informing him about Frittelli's paper \cite{fritteli2007}. All three authors would like to thank Freddy Delbaen, Valentina Galvani, and an anonymous associate editor for their valuable advices and suggestions that helped improving the paper tremendously. The three authors are very grateful for NSERC (the Natural Sciences
and Engineering Research Council of Canada) that supported financially this research through the Grant G121210818.



\end{document}